\documentclass[runningheads]{llncs}
\usepackage[T1]{fontenc}
\usepackage{amssymb}
\usepackage{csquotes}

\usepackage{wrapfig}
\usepackage{multicol}

\DeclareMathAlphabet{\mathcal}{OMS}{cmsy}{m}{n}
\SetMathAlphabet{\mathcal}{bold}{OMS}{cmsy}{b}{n}
\newcommand\bigO{\mathcal{O}}

\newcommand\givespace{\vspace{0pt}}

\usepackage[english,linesnumbered,ruled,vlined,noend]{algorithm2e}
\usepackage{listings}
\usepackage{listings-rust}

\usepackage{cite}
\usepackage{hyperref}
\hypersetup{
  hidelinks,
  colorlinks=true,
  allcolors=black,
  pdfstartview=Fit,
  breaklinks=true
}
\usepackage[capitalise,nameinlink]{cleveref}

\crefname{section}{Sect.}{Sect.}
\Crefname{section}{Section}{Sections}
\crefname{listing}{List.}{List.}
\crefname{listing}{Listing}{Listings}
\Crefname{listing}{Listing}{Listings}
\crefname{lstlisting}{Listing}{Listings}
\Crefname{lstlisting}{Listing}{Listings}

\begin{document}

\title{PermRust: A Token-based\\Permission System for Rust}
\author{Lukas~Gehring\inst{1} \and%
Sebastian~Rehms\inst{1} \and%
Florian~Tschorsch\inst{1}}
\institute{Technische Universität Dresden, 01062 Dresden, Germany
\email{\{lukas.gehring,sebastian.rehms,florian.tschorsch\}@tu-dresden.de}}
\maketitle
\begin{abstract}
Permission systems which restrict access to system resources are a well-established technology in operating systems, especially for smartphones.
However, as such systems are implemented in the operating system they can at most manage access on the process-level.
Since moderns software often (re)uses code from third-parties libraries, a permission system for libraries can be desirable to enhance security.
In this short-paper, we adapt concepts from capability systems
building a novel theoretical foundation for permission system at the level of the programming language.
This leads to \emph{PermRust}, a token-based permission system for the Rust programming language as a zero cost abstraction on top of its type-system.
With it access to system resources can be managed per library.
\end{abstract}


\section{Introduction}\label{Introduction}
Managing access to sensitive resources is a crucial requirement to secure systems.
This holds especially when executing code from third-parties.
For example, modern operating systems (like Android and iOS) implement permission systems to restrict apps from using resources (like camera, filesystem, GPS) without a permission from the user.
This partially realizes the \emph{principle of least privilege}, stating that a subject should only be able to access the resources that are necessary for its legitimate purpose.
Such systems are implemented at OS-level managing resources per process.
This leaves an open space for systems which restricts more granular subjects like libraries.
Since in modern software development, developers build their apps on top of various third-party libraries, distributed in open repositories (like \verb|cargo| or \verb|npm|), restricting access to resources for these libraries can help mitigate against attacks on (or from) those libraries (so-called \emph{supply-chain-attacks}).

The Pony programming language~\cite{Pony} utilizes such a permission system,
known as object capabilities.
Pony however has a low adoption rate and steep learning curve, because it uses a complex system of reference capabilities to ensure memory safety.
While the scientific literature discusses the theoretical foundations of the latter in detail~\cite{pony_refcapa}, its permission system receives less attention.

In this short paper, we introduce \emph{PermRust}, a token-based permission system for the Rust programming language.
PermRust allows developers to easily restrict access to system resources per library.
As it utilizes Rust's type system to enforce the restrictions, PermRust only increases the compile time and has no run time costs.
The system can be used to manage the permissions of third-party libraries, lowering the risk of supply chain attacks.

We introduce PermRust as a concept to foster discussions on permission systems in modern programming languages.
We provide a proof-of-concept implementation, available in our accompanying Git repository~\cite{gitrepo}.
Please note that we consider an evaluation of PermRust as future work and limit this paper to a discussion of the approach's limitations.
Our contributions are as follows:
\begin{itemize}
	\item We lay a novel theoretical foundation to restrict system resources per function in modern software development.
		For this, we leverage established access control concepts such as access matrices and capability systems~(\cref{sec:methodology}).
	\item We employ this method to develop PermRust as a Rust-based proof of concept~(\cref{implementing}).
\end{itemize}
We conclude the paper with a discussion on limitations of our approach~(\cref{chap:evaluation}) and outline areas of future work.%

\section{Related Work}\label{relatedwork}

In this paper, we sketch a framework for app-developers which mitigates supply chain attacks by restricting the access of third-party libraries to system resources.
A lot of related work in this area has been done around the Android operating system~\cite{LibCage,Flexdroid,Compac}.
Other techniques to restrict I/O access on OS-level are SELinux~\cite{SELinux}, AppArmor~\cite{AppArmor}, seccomp, and pledge~\cite{seccomp}.
These techniques are blind to program components and do not operate on a per-library level.
Existing languages, which implement such capabilities to restrict access to resources, are Joule~\cite{Joule}, E~\cite{E}, and Pony~\cite{Pony}.
These languages, however, are not designed with supply chain attacks in mind and are constrained to a single programming paradigm, i.e., Dataflow, object-oriented, and actor-based programming.
In addition, they come with runtime penalties (e.g., Joule, E) or forbid certain functionalities (e.g., no global variables in Pony).
PermRust brings permission systems to a popular programming language with a sizable ecosystem.
A lot of work has been put into techniques to mitigate supply chain attacks,
including statistical program analysis~\cite{Garrett,Pfretzschner,duan_towards_2021} or precaution against human errors~\cite{Zimmermann}.
\section{Access Management}\label{sec:methodology}
In this section, we develop the theoretical foundation for managing access to system resources per function.
In particular, we introduce and adapt the well-established Access Matrix Model for our use-case and argue that the model can best be instantiated with a capability based approach.

\subsection{Adapting the Access Matrix Model}\label{terms?}
The \emph{access matrix model}~\cite{lampson1971protection} is a simple and well-established computer security model to formalize a security policy.
For its definition, we use the formalization by~\cite[p. 264]{eckert2023sicherheit}.
\begin{definition}[Access Matrix Model]
Let $S_t$ be a finite set of \emph{subjects} and $O_t$ a finite set of \emph{objects} at time~$t$.
Let $R$ be a finite set of access rights.
The matrix~$M_t \in R^{|S_t | \times |O_t |}$ is the \emph{access matrix} over $S, O$, and $R$ at time~$t$.
$M_t (s,o) = \{r_1 , \dots r_n \} $ is the set of access rights a subject~$s$ has to the object~$o$ at time~$t$.
\end{definition}

As described our security system should manage the access of libraries to system I/O.
Hence, the sets of subjects correspond to the functions of a program and the sets of objects correspond to the systems I/O-interfaces.
High-level languages provide library functions as entry points for I/O operations.
We model the set~$O$ as this specific subset of the program functions.
With that, we only need to consider one type of access right: $R = \{call\}$.
E.g a function $f$ having read-access is modelled by having the right to \emph{call} the functions $f'$ of the standard library, which implement read-operations ($M(f,f')= \{call\}$).
The following definitions formalize security policies which lead to an algorithm that checks if a program complies with a given policy.
\begin{definition}[Permission Matrix for a Program]\label{def:permMatrix}
Let $F$ be the set of functions of a program~$\Pi$.
For $O \subset F$, an access matrix $M$ over $F, O$ and \{$call$\} is a \emph{permission matrix over $\Pi$}.
\end{definition}
In a program, each function has a caller that (per definition of the permission matrix) also has permissions.
Because it is not immediately clear how the permissions of the caller trickle-down to the functions further up the call stack we need to define some properties connecting the access rights of functions.
We now define a subtree of the abstract syntax tree of a program which enables us to formally talk about the caller-callee relations of functions.
For the construction we use edge contraction a common operation in graph theory, which removes edges $e = (v,w)$ from a graph and merges the two vertices $v$ and $w$.
\begin{definition}[Abstract Function Tree]
Let $T=(V,E)$ be the abstract syntax tree of $\Pi$ with $F \subseteq V$.
Let $E' = \{(v,w) \in E | w \notin F\}$, be the set of all edges not ending on a function vertex.
The tree $T' = T / E'$, which results from $T$ by contracting all edges in $E'$ and naming them after the target of the original edge, is called the \emph{abstract function tree~(AFT) of $\Pi$}.
\end{definition}
\begin{definition}[Permission Respecting]
A tree $T=(V, E)$ is \emph{permission respecting regarding an access matrix M} if
$\forall (p,o) \in E, o \in O: M(p,o) = \{call\}$.
A program $\Pi$ is \emph{permission respecting w.r.t. an access matrix $M$}, if its AFT is permission respecting w.r.t. $M$.
\end{definition}
\begin{definition}[Privilege Escalation Free]
A tree $T=(V, E)$ is \emph{privilege escalation free over M} if
$\forall (p,c) \in E, \forall o \in O: M(p,o) \supseteq M(c,o)$.
A program $\Pi$ is \emph{privilege escalation free over an access matrix $M$}, if its AFT is privilege escalation free over $M$.
\end{definition}
Note that when $M$ is a permission matrix for a program $M(s,o)$ can only be  $\{call\}$ or $\emptyset$.
This simplifies the condition of the definition to 
$\forall (p,c) \in E, \forall o \in O: M(c,o) =\{call\} \Rightarrow M(p,o) = \{call\}$.
\begin{theorem}
	There exists an algorithm which outputs whether $T$ is permission respecting regarding $M$ and privilege escalation free over $M$ which halts after $\bigO(|E|\cdot|O|)$ read accesses to $M$.
\end{theorem}
\begin{proof}
\Cref{alg:permissioncheck} has the required properties.
The algorithm is correct, since for a permission respecting and privilege escalation free tree the conditions in line~\cref{algline:if1} and~\cref{algline:if2} are by definition never fulfilled and therefore returns \texttt{true}.
On the other hand, if a tree does not hold both of those characteristics, there is an edge where one of the condition fails and the algorithm returns \texttt{false}.
For every edge, the matrix needs to be accessed $t \leq 1 + 2|O|$ times.
This results in a runtime of $\bigO(|E| (1 + 2|O|)) = \bigO(|E|\cdot |O|)$.
\begin{algorithm}[tb]\footnotesize
        \SetAlgoLined%
				\caption{Permission Check}\label{alg:permissioncheck}
    \KwIn{Tree $T = (V,E)$, permission matrix $M$ over a program $\Pi$ with $V \subseteq F$ and $O \subseteq F$}
		\KwOut{Is $T$ permission respecting regarding $M$ and privilege escalation free over $M$?}
        \ForEach{$(p,c) \in E$}{
					\If{$c \in O$}{
						\If{$M(p,c) \neq \{call\}$\label{algline:if1}}{
							\Return{false} \tcp*[f]{not permission respecting}}
						}
					\ForEach{$o \in O$}{
						\If{$M(p,o) \subsetneq M(c,o)$\label{algline:if2}}{
							\Return{false \tcp*[f]{privilege escalation found}}
						}
					}
        }
				\Return{true}
\end{algorithm}
\end{proof}

In~\cref{implementing}, we show that the clever usage of types can lead to the indirect execution of this algorithm inside the type checker of a strongly typed language.

\subsection{Adapting Capabilities}\label{capa?}
\emph{Access control lists (ACL)} and \emph{capabilities}~\cite{lampson1971protection} are the most common implementations of access matrices in modern systems~\cite[p. 635]{eckert2023sicherheit}.
The main difference between the two models is where the permissions are stored.
With ACL, for every subject~$s$ a set of access rights~$M(s,o)$ is saved beside an object~$o$.
The capability model implements this the other way around:
every subject~$s$ has a set of capabilities holding the access rights~$M (s,o)$ for every object~$o$.

Miller et al.~\cite{Myths} compare different approaches to implement such ACL and capability systems.
For our use case, only a few criteria to categorize different capability systems are important.
Especially, it is unimportant to dynamically change the permission matrix.
ACL are unsuitable for our use case, since our objects (the system resources) exist independently and on a lower abstraction layer (the OS or the compiler) than the functions, which change depending on the program.
Furthermore, we need our framework to not permit ambient authority, which is \textquote{authority that is exercised, but not [explicitly] selected by its user}~\cite{Myths}.
Ambient authorities would subvert the goal of our model to make the access to I/O-resources explicit.
In addition, such authorities would require additional data structures, where the permissions are saved and implicitly queried when calling a function restricted by the permission model.

Consequentially, we have two possible techniques to implement a capability-system to mitigate supply chain attacks: the \emph{capability-as-keys} model and \emph{object capabilities}.
In the \emph{capability-as-keys} model, the subject needs to provide a correct key (or token) to access an object.
The keys need to be unforgeable, copy-able and only access the specific resource it was designed for.
Permission to unlock an object can only be obtained from another entity which already had the corresponding key (or by being the special \emph{root}-subject, which holds all keys from the start).
In contrast to doors in the real world, keys only open doors for one-time entry and for the holder of the key only.
\emph{Object capabilities} are similar, but they cut out the middleman by connecting authority with designation directly.
Specific for our example that would mean, that the structures in the standard library which represent, e.g., file-descriptors or network-sockets need to have the same properties as keys in the capabilities-as-keys model.
Since the token-based approach is generally less disruptive to the common workflow used when developing with Rust, we focus on an implementation using this model.

We now show that safe Rust fulfills the criteria for a capability-aware programming language~\cite{E}.
The first property a language must fulfill is memory safety.
Without this feature, a token would be forgeable, meaning that keys could be construed at will.
In a memory unsafe language any line of code can call the constructor of any given key or copy the key from another place in memory.
Since the borrow checker used in Rust ensures memory safety, this criterion is fulfilled.
The second criteria for a capability-system is encapsulation.
It means that \textquote{you cannot reach inside an object for its instance variables}~\cite{E}.
This is mostly important, if we use object capabilities and can be achieved using visibility in Rust.
For a system to support capabilities, we need to restrict global variables in a way which ensures that authority can only be obtained explicitly.
This means, that keys cannot be saved in static variables since they could be used to distribute tokens to unauthorized players.
The general way to do this is to only allow immutable global state, which is not controversial in modern development, since global variables have been considered harmful since 1973~\cite{wulf1973global}.
For this reason, they are also heavily restricted in safe Rust.
Therefore, as long as libraries are not using unsafe Rust, the requirements are fullfilled.

\section{Proof of Concept}\label{implementing}
In this section, we outline a proof-of-concept implementation of a capability-based permission system based on Rusts type system, which we call \emph{PermRust}.
We first focus on realizing a way to label functions, which communicate the I/O-operations the function uses under any condition.
Secondly, we sketch how namespaces can be imported with annotated permissions.
Finally, we bring both together and describe, how we can make sure, that functions can only be called if the namespace is imported with the necessary permissions.

\subsection{Labeling Functions with Permissions in Rust}\label{sec:impl_instr}
Our first goal is to label Rust functions with the permissions they need.
As discussed, we will use the capability-as-key model for this and require every I/O-accessing function to take corresponding tokens (or \emph{keys}) as arguments.
The code that implements a token corresponding to read permission is shown in~\cref{lst:token}.
We call such types \emph{token types}.
Since token types use no runtime memory, they are a form of \emph{zero-cost abstraction}, which means that the permission check is entirely done at compile time, introducing no runtime costs.
This is possible because Rust is a strongly typed language, where a missing or wrong type leads to an error at compile time.
A key property in capability models is that they cannot be constructed at will.
We use Rusts visibility feature to ensure that objects with a token type cannot be constructed outside its namespace (called modules in Rust).
The creation of a struct requires, that all fields of the struct are public.
As such, trying to construct a token type outside the token module, as depicted in~\cref{lst:token_creation_error}, fails with an error message.
In a full-fledged implementation of PermRust the \texttt{token} module would be a part of the standard library and contain multiple different token types, since every library developer would need to use these types to write a function which access I/O.

\begin{figure}[t]
	\noindent\begin{minipage}[t]{.45\textwidth}
\begin{lstlisting}[language=Rust,frame=tlrb, caption={A token type in Rust. A struct containing a single field holding an empty tuple.}, label={lst:token}]
mod token {
    pub struct ReadPerm(());
}
\end{lstlisting}
\begin{lstlisting}[language=Rust,frame=tlrb, caption={Unsuccessful creation of token types yields an error message because the field is private.}, label={lst:token_creation_error}]
let token = token::ReadPerm(());
\end{lstlisting}
\begin{lstlisting}[language=Rust,frame=tlrb,caption={A library function signature with read access to the file system.},label={lst:libFunction}]
/// # Permissions
/// - "ReadPerm"
pub fn read_something(
    f: std::fs::File,
    read_token: &token::ReadPerm,
)
\end{lstlisting}
	\end{minipage}\hfill
	\begin{minipage}[t]{.45\textwidth}
\begin{lstlisting}[language=Rust, frame=tlrb, caption={A Rust implementation of the standard \texttt{read}-function requiring a correct token.}, label={lst:proxy}]
mod proxy_std {
    pub struct FileProxy(
        pub std::fs::File,
    );
    impl FileProxy {
        pub fn read(
            &mut self,
            buf: &mut [u8],
            _: &token::ReadPerm,
        ) -> Result<usize> {
            self.0.read(buf)
        }
    }
}
\end{lstlisting}
\begin{lstlisting}[language=Rust,frame=tlrb,caption={Usage of \texttt{lib\_func} to generate a function requesting tokens.},label={lst:useAddArguments}]
#[lib_func("ReadPerm")]
pub fn read_something(
    f: std::fs::File,
)
\end{lstlisting}
	\end{minipage}
\end{figure}

Next, we ensure that the structs with the correct token types are actually required to perform the corresponding I/O-operation.
The only way to do I/O-operations in safe Rust is via calls to the standard library, which we therefore have to modify.
A possible solution for read permissions which puts itself in front of the original standard library is shown in~\cref{lst:proxy}.
In Rust, one interface of the standard library to perform a read operation is to call the \texttt{read()} method from a \texttt{std::fs::File} object,
which represent a file descriptor.
\Cref{lst:proxy} uses the new type pattern to introduce a proxy type that allows us to write a new interface for the \texttt{File} type without the need to change to original.
The read function of the \texttt{FileProxy} struct takes the same arguments as the original read function plus a \texttt{ReadPerm} token.
Since the only purpose of the token is to ensure that the caller possesses such a token, it is not used by the function.
The other arguments and the return value of the underlying \texttt{File} object's read function are unmodified.
Because of Rust's \emph{zero-cost abstraction} there is no runtime penalty for this rewrite.

Since the standard library now expects the correct token to be called and there is no way to generate these tokens, library-functions need to require them from their caller.
\cref{lst:libFunction} shows the signature of such a function, which requires a \texttt{ReadPerm} token.
It is also annotated with special comments, which the tool \texttt{rustdoc} can use to generate documentation for Rust projects.
In order to make the required permission even more apparent, we suggest auto generating them using Rust's procedural macros.
Such a macro named \texttt{lib\_func} would allow the developer to write to code in~\cref{lst:useAddArguments} to generate~\cref{lst:libFunction}.

Together with the modified standard library, it is ensured that no function can perform any I/O operation in safe Rust without having the correct access tokens.

\subsection{Permission-aware Importing in Rust}\label{impl_permAware}
While our implementation forces library developers to annotate their functions with permissions, it is not clear which and when tokens are initially generated.
We suggest setting the permissions of packages in the \texttt{Cargo.toml} file, where dependencies are listed in Rust.
Furthermore, we build special \texttt{app\_} functions, which work as entry points for app developers.
These entry points should only work if the correct permission for their packages are set in \texttt{Cargo.toml}.
The construction of the \texttt{app\_} functions can be automated using macros.
An implementation can be found in our git repository~\cite{gitrepo}.

As a first step to make Rust permission-aware, we suggest that all packages have special kinds of \texttt{features}\footnote{%
	The \texttt{feature} mechanism of Cargo and Rust is the way conditional compilation is implemented in Rust.
} called \emph{permissions features}, which represent access rights to I/O operations.
These should align with the permissions represented by the token types.
This can be used, to prohibit the compilation of functions executing more I/O operations as desired by the developer.
We can also use the conditions to create tokens when the corresponding permission-features are enabled.

\Cref{lst:app_libFunction_with_bundle} shows an \texttt{app\_} function, which internally generates the correct tokens required to do I/O-operations.
It prevents the developer from using a token with more permissions than defined in \texttt{Cargo.toml}.
The code in Lines~2--13 introduces a new token called \texttt{localPerm}, which should be usable as a token type depending on the permission set.

\begin{lstlisting}[language=Rust,float, %
caption={A new interface to \texttt{read\_something}, which does not require the generation of tokens.
Line 2--13 generates a local token which can be used as \texttt{ReadPerm}.},
label={lst:app_libFunction_with_bundle}]
pub fn app_read_something(f: std::fs::File) {
    struct localPerm(());
    impl localPerm {
        const fn new() -> Self {
            localPerm(())
        }
    }
    #[cfg(feature = "ReadPerm")]
    impl std::convert::AsRef<token::ReadPerm> for localPerm {
        fn as_ref(&self) -> &token::ReadPerm {
            unsafe { std::mem::transmute(self) }
        }
    }
    read_something(f, localPerm::new().as_ref())
}
\end{lstlisting}

\subsection{Implementation}
We now combine the functionality in a way that only one macro is needed.
We will call this macro \texttt{permissions}. an implementation can be found in our git repository~\cite{gitrepo}.
\texttt{permissions} is intended to be used on functions such as \texttt{read\_something} in~\cref{lst:useAddArguments}.
These functions directly or indirectly perform some kind of I/O-operation and do not have tokens as inputs.
We therefore generate two functions with \texttt{permissions}.
The first function adds the token-inputs and documentation (as in~\texttt{lib\_func}).
This function should only be compiled if the correct permission-features are set.%
\footnote{Conceptionally this is not necessary, since the function can only be called by a function, which has the correct tokens.
However, adding the conditional check can make the executable smaller.
}

The other function, which should be generated, is the \texttt{app\_read\_something} function~(\cref{lst:app_libFunction_with_bundle}).
However, since cargo-features are not transitive, a dependency which is pulled without \texttt{ReadPerm} can, for example, still rely on another package with \texttt{ReadPerm} enabled.
As such it could use the \texttt{app\_} functions of its dependency, leading to privilege escalation.
By introducing a new feature call \emph{direct-dependency}, which is only set for packages that are direct dependencies of the project, we can prevent this scenario.
Currently, Cargo does not implement the functionality to automatically flag direct dependencies and to ensure, that the feature is not set by other dependencies.
However, since Cargo constructs the whole dependency-tree it is reasonable to assume, that such a feature could be implemented.

We can now construct the fully functional \texttt{permission} macro.
The steps necessary for the macro are: cloning to function, prefixing the cloned function name with \texttt{app\_}, adding the call to the original function, adding the conditional compilation arguments.

\section{Discussion and Limitations}\label{chap:evaluation}
In the following, we analyze and discuss the limitations of PermRust
with respect to cost, permission granularity, feature unification, customization and attacker models.
We furthermore provide ideas on how these limitations could be addressed in the future.

\givespace\paragraph{Costs}\label{costs}
Since the token types occupy zero memory per definition, they do not exist at runtime and therefore do not extend the runtime cost of the application.
However, the development-cost of third-party libraries increase slightly, since every function with I/O-access needs to be correctly annotated by the developer and tokens need to be provided when calling such a function.
The \texttt{permission} macro should minimize the workload, making it possible to write function in PermRust almost exactly as in Rust.
Using the macro we have to pay a high price in respect to compile time, since type-checking and executing of macro code is done by the compiler.
Because the security of the whole system stems from the fact that the standard library is a trust anchor, language maintainers need to spend more time on designing and implementing the API in a capability-compatible way.
For this reason, the probability that mainline Rust will implement the discussed permission model is rather small.

\givespace\paragraph{Granularity}\label{granularity}
To ensure permission-aware importing, we matched the permissions required from a function with the permission set for a package via the constructions of \texttt{app\_} functions.
While this makes sense from the perspective of the developers of the library, a more granular approach could be more usable.
The permission would then not be defined in the \texttt{Cargo.toml} file but when importing different paths in the source code using Rusts \texttt{use} keyword.

\paragraph{Feature Unification}\label{featureuni}
Feature unification is a feature of cargo, which is used when a package is present multiple times in the dependency tree of a project~\cite{cargo}.
This can happen if for example a project has the direct dependencies~\texttt{LibA} and \texttt{LibB} and \texttt{LibB} also depends on \texttt{LibA}.
Since all dependencies can be imported with different features, cargo needs to decide which features it should set for \texttt{LibA}.
The solution of the package manager is to calculate the union of the desired features and compile \texttt{LibA} with them.

PermRust relies on the property, that the \textquote{direct-dependency} feature is only set for direct dependencies.
Regardless, since \texttt{LibA} is now a direct- \emph{and} indirect-dependency,
the feature will be activated and enable \texttt{LibB} to call the \texttt{app\_} functions of \texttt{LibA}.
This could lead to privilege escalation, because \texttt{LibB} can call the functions regardless of its permissions.
The real-world effect of this vulnerability is unclear, since the author of \texttt{LibB} cannot make assumptions on the dependencies of other projects.
Furthermore, it is unlikely that the developer of \texttt{LibB} would use an \texttt{app\_} version by accident.

\givespace\paragraph{Customization}\label{Customizability}
PermRust is tightly connected with the standard library and therefore with the system calls of the underlying system.
This logically leads to the restriction, that permissions can only be given in full and not customized regarding concepts not known to the underlying system.
For example, it would be useful to be able to give access rights only for certain paths of the file system or to allow TCP-connections only to a specific IP-range.
While different standard library functions for common customization could be created, they would need to check the custom conditions at runtime.

A full object capability system as described in~\cref{sec:methodology}, could allow custom permissions.
In contrast to the capability-as-keys model, object capabilities are not tokens which allow access to resources but the objects representing the resources themselves.
For example library functions interacting with the file system would need to get a \texttt{File} object from its caller, and could only interact with the file represented by that object.
This object would need to be made unforgeable, meaning that the \texttt{open} and \texttt{create} methods could only be executed where token types are constructed in PermRust.
This would of course interrupt the development flow of common Rust developers even more, than the capability-as-keys model.
Since in this system, all resources of a program would need to be obtained at the beginning of the program, startup times would increase rapidly.

\givespace\paragraph{Attacker Models}\label{Attacker Models}
It is clear that all security mechanism of PermRust can easily be circumvented by using unsafe Rust to construct an arbitrary token type at will.
Despite that, PermRust can still be useful to mitigate real supply chain attacks.
Specifically the ones which stem from badly designed libraries.
Such libraries provide functions which (in some corner-cases) access system resources without making the access obvious to the application developer.
An example of such an attack is Log4Shell, where a popular logging library, in some cases, contacted a remote server leading to remote code execution. A description of the attack can be found in~\cite{chowdhury2022better}.

\section{Conclusion}

In this work, we addressed the principles of least privilege in the context of library usage in modern software development.
Our proof of concept involves a framework that clearly defines I/O-operations
and binds permissions to libraries, limiting function execution.
Using access matrices and capability systems,
we developed a capability-secure programming language called \emph{PermRust}
using Rust's procedural macros.
PermRust employs a capability-as-key model,
restricting system I/O access to token owners.
We demonstrated the feasibility of this model in Rust
by providing a proof of concept.

\bibliographystyle{splncs04}
\bibliography{bib.bib}

\begin{thebibliography}{10}
\providecommand{\url}[1]{\texttt{#1}}
\providecommand{\urlprefix}{URL }
\providecommand{\doi}[1]{https://doi.org/#1}

\bibitem{Joule}
Agorics, I.: Joule: Distributed application foundations (1995),
  \url{http://erights.org/history/joule/}

\bibitem{seccomp}
Anderson, J.: Sandboxing techniques. FreeBSD Journal  (2017)

\bibitem{chowdhury2022better}
Chowdhury, P.D., Tahaei, M., Rashid, A.: Better call saltzer \& schroeder: A
  retrospective security analysis of solarwinds \& log4j. arXiv:2211.02341
  (2022)

\bibitem{duan_towards_2021}
Duan, R., Alrawi, O., Kasturi, R.P., Elder, R., Saltaformaggio, B., Lee, W.:
  Towards measuring supply chain attacks on package managers for interpreted
  languages. In: Proceedings 2021 NDSS Symposium. Internet Society (2021)

\bibitem{eckert2023sicherheit}
Eckert, C.: IT-Sicherheit: Konzepte--Verfahren--Protokolle. de Gruyter (2023)

\bibitem{Garrett}
Garrett, K., Ferreira, G., Jia, L., Sunshine, J., K{\"a}stner, C.: Detecting
  suspicious package updates. In: 2019 IEEE/ACM 41st ICSE-NIER. pp. 13--16.
  IEEE (2019)

\bibitem{gitrepo}
Gehring, L.: Code for permrust: A token-based permission system for rust,
  \url{https://git.sr.ht/~lgehr/token_based_permission_system_code}

\bibitem{AppArmor}
Gruenbacher, A., Arnold, S.: Apparmor technical documentation (2007)

\bibitem{lampson1971protection}
Lampson, B.W.: Protection. SIGOPS Oper. Syst. Rev.  \textbf{8}(1),  18–24
  (jan 1974)

\bibitem{Myths}
Miller, M.S., Yee, K.P., Shapiro, J.: Capability myths demolished. Tech. rep.,
  Johns Hopkins University Systems Research (2003)

\bibitem{Pfretzschner}
Pfretzschner, B., ben Othmane, L.: Identification of dependency-based attacks
  on node. js. In: Proceedings of the 12th International Conference on
  Availability, Reliability and Security. pp.~1--6 (2017)

\bibitem{Pony}
{Pony Developers}: Pony, \url{https://www.ponylang.io}

\bibitem{Flexdroid}
Seo, J., Kim, D., Cho, D., Shin, I., Kim, T.: Flexdroid: Enforcing in-app
  privilege separation in android. In: NDSS (2016)

\bibitem{pony_refcapa}
Steed, G., Drossopoulou, S.: A principled design of capabilities in pony

\bibitem{E}
Stiegler, M.: \url{http://www.skyhunter.com/marcs/ewalnut.html#SEC41}

\bibitem{cargo}
{The Cargo Team}: The cargo book, \url{https://doc.rust-lang.org/cargo/}

\bibitem{LibCage}
Wang, F., Zhang, Y., Wang, K., Liu, P., Wang, W.: Stay in your cage! {A} sound
  sandbox for third-party libraries on android. In: Computer Security - 21st
  {ESORICS}. pp. 458--476. Springer (2016)

\bibitem{Compac}
Wang, Y., Hariharan, S., Zhao, C., Liu, J., Du, W.: Compac: Enforce
  component-level access control in android. In: Proceedings of the 4th ACM
  Conference on Data and Application Security and Privacy. pp. 25--36 (2014)

\bibitem{SELinux}
Wikberg, M.: Secure computing: Selinux (2007),
  \url{https://citeseerx.ist.psu.edu/document?repid=rep1&type=pdf&doi=242132951b3157f1d887d507b1c0289fd27e16eb}

\bibitem{wulf1973global}
Wulf, W., Shaw, M.: Global variable considered harmful. ACM Sigplan notices
  \textbf{8}(2),  28--34 (1973)

\bibitem{Zimmermann}
Zimmermann, M., Staicu, C.A., Tenny, C., Pradel, M.: Small world with high
  risks: A study of security threats in the npm ecosystem. In: USENIX security
  symposium. vol.~17 (2019)

\end{thebibliography}

\end{document}